\documentclass{article}
\usepackage{amsmath}
\usepackage{amsthm}
\usepackage{amsfonts}
\usepackage{amssymb}
\usepackage{graphicx}
\usepackage{caption}
\usepackage{listings}
\usepackage{paralist}
\usepackage{tikz}
\usepackage[ruled,commentsnumbered,linesnumbered,noend]{algorithm2e}
\usepackage{url}
\usepackage{array}
\usepackage{floatrow}
\usepackage{lipsum}
\usepackage{xcolor}
\usepackage{color}
\usepackage{algorithmic}
\usepackage{booktabs}
\usepackage{multirow}
\usepackage{dsfont}
\usepackage{epstopdf}
\usepackage{subcaption}
\usepackage{balance}
\usepackage{makecell}
\usepackage{pifont}

\newcommand{\cmark}{\ding{51}}%
\newcommand{\xmark}{\ding{55}}%

\newcommand{\Simulink}{{Simulink\textsuperscript{\textregistered}}}
\newcommand{\Mathworks}{{Mathworks\textsuperscript{\textregistered}}}
\newcommand{\toolname}{{\sc DryVR}}
\newcommand{\edgelab}{{\sc elab}}
\newcommand{\vertlab}{{\sc vlab}}
\newcommand{\modemap}{{\sc lmap}}
\newcommand{\simulator}{{\sc sim}}
\newcommand{\dom}{{\mathit dom}}
\newcommand{\fstate}{\mathit{fstate}}

\newcommand{\ltime}{\mathit{ltime}}
\newcommand{\lstate}{\mathit{lstate}}
\newcommand{\lmode}{\mathit{lmode}}
\newcommand{\fmode}{\mathit{fmode}}
\newcommand{\reach}[1]{\relax\ifmmode {\sf Reach}_{#1} \else ${\sf Reach}_{#1}$\fi}
\newcommand{\Reach}[1]{\relax\ifmmode {\sf Reach}_{#1} \else ${\sf Reach}_{#1}$\fi}
\newcommand{\post}[4]{\relax\ifmmode {\sf Post}_{#1}(#2,#3,#4) \else ${\sf Post}_{#1}(#2,#3,#4)$\fi}
\newcommand{\posti}[1]{\relax\ifmmode {\sf Post}_{#1} \else ${\sf Post}_{#1}$\fi}
\newcommand{\paths}[1]{\relax\ifmmode {\sf Paths}_{#1} \else ${\sf Paths}_{#1}$\fi}
\newcommand{\pathof}[1]{\relax\ifmmode {\sf path}({#1}) \else ${\sf path}({#1})$\fi}
\newcommand{\execs}[1]{\relax\ifmmode {\sf Execs}_{#1} \else ${\sf Exec}_{#1}$\fi}
\newcommand{\traces}[1]{\relax\ifmmode {\sf Trace}_{#1} \else ${\sf Trace}_{#1}$\fi}
\newcommand{\GraphReach}{\mathit{GraphReach}}
\newcommand{\VerInit}{\mathit{VerInit}}

\newcommand{\Lmode}[1]{\relax\ifmmode {\sf {#1}} \else ${\sf {#1}}$\fi}

\input{prelude1}
\newfont{\mycrnotice}{ptmr8t at 7pt}
\newfont{\myconfname}{ptmri8t at 7pt}

\newcommand{\computeRT}{\mathit{ReachComp}}

\theoremstyle{plain}
\newtheorem{theorem}{Theorem}[section]
\theoremstyle{definition}
\newtheorem{definition}[theorem]{Definition}
\theoremstyle{example}

\newtheorem{proposition}[theorem]{Proposition}

\newtheorem{remark}[theorem]{Remark}

\begin{document}
\title{\toolname: Data-driven verification and compositional reasoning for automotive  systems}
\author{Chuchu Fan \and Bolun Qi \and Sayan Mitra \and Mahesh Viswanathan\\University of Illinois at Urbana-Champaign}
\date{}
\maketitle
\begin{abstract}
	We present the \toolname\ framework for verifying  hybrid control systems that are described by a combination of a black-box simulator for trajectories and a white-box transition graph specifying mode switches. 
	The framework  includes (a) a probabilistic algorithm for learning  sensitivity of the continuous trajectories  from simulation data, 
	(b) a bounded reachability analysis algorithm that uses the learned sensitivity, and 
	(c) reasoning techniques based on simulation relations and sequential composition, that enable verification of complex systems under long switching sequences, from the reachability analysis of a simpler system under shorter sequences. 
	 We demonstrate the utility of the framework by verifying a suite of automotive benchmarks that include powertrain control, automatic transmission, and several autonomous and ADAS features like automatic emergency braking, lane-merge, and auto-passing controllers.
\end{abstract}
%

\section{Introduction}
\label{sec:intro}

The starting point of existing hybrid system verification approaches
is the availability of nice mathematical models describing the
transitions and trajectories. This central conceit severely restricts
the applicability of the resulting approaches. Real world control
system ``models'' are typically a heterogeneous mix of simulation
code, differential equations, block diagrams, and hand-crafted look-up
tables. Extracting clean mathematical models from these descriptions
is usually infeasible.
At the same time, rapid developments in Advanced Driving Assist Systems (ADAS), autonomous vehicles, robotics, and drones now make the need  for effective and sound verification algorithms stronger than ever before. 
The \toolname\ framework  presented in this paper aims to narrow the gap between sound  and practical verification for control systems.
\vspace{-10pt}
\paragraph{Model assumptions}
 Consider an ADAS feature like automatic emergency braking system (AEB). The high-level logic deciding the timing of when and for how long the brakes are engaged after an obstacle is detected by sensors is implemented in a relatively clean piece of code and this logical module can be seen as a {\em white-box\/}. In contrast, the dynamics of vehicle itself, with  hundreds of parameters, is more naturally viewed as a {\em black-box\/}. That is, it can be simulated or tested with different initial conditions and  inputs, but it is nearly impossible to write down a nice mathematical model.

 The empirical observation motivating this work is that many  control systems, and especially automotive systems, share this  combination of  white  and  black boxes (see other examples in Sections~\ref{ex:powertrain}, \ref{ssec:adas}, and \ref{ssec:gear}).
 In this paper, we view hybrid systems as a combination of a white-box that specifies the mode switches and a black-box that can  simulate the continuous evolution in each mode. 
 Suppose the system has a set of modes $\L$ and $n$ continuous variables. The mode switches are defined by a {\em transition graph} $G$ which is a directed acyclic graph (DAG) whose vertices and edges define the allowed mode switches and the switching times. The black-box is a set of trajectories $\TL$ in $\reals^n$ for each mode in $\L$. 
 We do not have a closed form description of $\TL$, but instead, we have a {\em simulator\/}, that can generate sampled data points on individual trajectories for a given initial state and mode.  
Combining a transition graph $G$, a set of trajectories $\TL$, and a
set of initial states in $\reals^n$, we obtain a hybrid system for
which executions, reachability, and trace containment can be defined
naturally.

We have studied a suite of automotive systems such as powertrain
control~\cite{jin2014powertrain}, automatic transmission
control~\cite{Matlab_trans}, and ADAS features like automatic
emergency braking (AEB), lane-change, and auto-passing, that
are naturally represented in the above style. In verifying a lane
change or merge controller, once the maneuver is activated, the mode
transitions occur within certain time intervals. In testing a
powertrain control system, the mode transitions are brought about by
the driver and it is standard to describe typical driver classes using
time-triggered signals. Similar observations hold in other examples.

\vspace{-10pt}
\paragraph{Safety verification algorithm}
With black-box modules in our hybrid systems, we address the  challenge of providing guaranteed verification. Our approach is based on the idea of simulation-driven reachability analysis~\cite{fan2016automatic,DMV:EMSOFT2013,DuggiralaMV:2015c2e2}. For a given mode $\ell \in \L$, finitely many simulations of the trajectories of $\ell$ and a {\em discrepancy function\/} bounding the sensitivity of these trajectories, is used to over-approximate the reachable states.
For the key step of computing  discrepancy for modes that are now represented by  black-boxes, we introduce a probabilistic algorithm that learns the parameters of exponential discrepancy functions from simulation data. The algorithm transforms the problem of learning the parameters of the discrepancy function to the problem of learning a linear separator for a set of points in $\reals^2$ that are obtained from transforming the simulation data. 
A classical result in PAC learning, ensures that any such discrepancy function works with high probability for all trajectories. We performed dozens of experiments with a variety of black-box simulators and observed that 15-20 simulation traces typically give a discrepancy function that works for nearly 100\% of all simulations.
The reachability algorithm  for the hybrid system proceeds along the vertices of the transition graph in a topologically sorted order and this gives a sound bounded time verification algorithm, provided the learned discrepancy function is correct. 



\vspace{-10pt}
\paragraph{Reasoning}
White-box transition graphs in our modelling, identify the switching
sequences under which the black-box modules are exercised.  Complex
systems have involved transition graphs that describe subtle sequences
in which the black-box modules are executed. To enable the analysis of
such systems, we identify reasoning principles that establish the
safety of system under a complex transition graph based on its safety
under a simpler transition graph. We define a notion of forward
simulation between transition graphs that provides a sufficient
condition of when one transition graph ``subsumes'' another --- if
$G_1$ is simulated by $G_2$ then the reachable states of a hybrid
system under $G_1$ are contained in the reachable states of the system
under $G_2$. Thus the safety of the system under $G_2$ implies the
safety under $G_1$. Moreover, we give a simple polynomial time
algorithm that can check if one transition graph is simulated by
another.

Our transition graphs are acyclic with transitions having bounded
switching times. Therefore, the executions of the systems we analyze
are over a bounded time, and have a bounded number of mode
switches. An important question to investigate is whether establishing
the safety for bounded time, enables one can conclude the safety of
the system for an arbitrarily long time and for arbitrarily many mode
switches. With this in mind, we define a notion of sequential
composition of transition graphs $G_1$ and $G_2$, such that switching
sequences allowed by the composed graph are the concatenation of the
sequences allowed by $G_1$ with those allowed by $G_2$. Then we prove
a sufficient condition on a transition graph $G$ such that safety of a
system under $G$ implies the safety of the system under arbitrarily
many compositions of $G$ with itself.

\vspace{-10pt}
\paragraph{Automotive applications}
We have implemented these ideas to create the {\bf D}ata-d{\bf r}iven
S{\bf y}stem for {\bf V}erification and {\bf R}easoning (\toolname).
The tool is able to automatically verify or find counter-examples in a
few minutes, for all the benchmark scenarios mentioned above.
Reachability analysis combined with compositional reasoning, enabled
us to infer safety of systems with respect to arbitrary transitions
and duration.
%
%


\vspace{-10pt}
\paragraph{Related work}
Most automated verification tools for hybrid systems rely on analyzing
a white-box mathematical model of the systems. They include tools
based on decidablity
results~\cite{doty95,hh95,ck99,adm02,Dutertre04timedsystems,fre05},
semi-decision procedures that over-approximate the reachable set of
states through symbolic
computation~\cite{gm99,mt00,bt00,kv00,st00,Frehse:cav11,ariadne,flow},
using
abstractions~\cite{adi03,efhkost03-1,efhkost03-2,Henzinger_refinement,seg07,07-HSolver,dkl07,JKWB:HSCC:2007,HareFMSD,14-sttt,14-AGAR,15-PLC-CEGAR,HareTACAS16},
and using approximate decision procedures for fragments of first-order
logic~\cite{dreach}.
More recently, there has been interest in developing simulation-based
verification
tools~\cite{Julius:2007:RTG:1760804.1760833,SensitivityDM,donze2010breach,Kanade09,staliro-tool-paper,Fainekos:2009:RTL:1609208.1609591,DRJqest13,DuggiralaMV:2015c2e2}. Even though
these are simulation based tools, they often rely on being to analyze
a mathematical model of the system. The type of analysis that they
rely on include instrumentation to extract a symbolic trace from a
simulation~\cite{Kanade09}, stochastic optimization to search for
counter-examples~\cite{staliro-tool-paper,Fainekos:2009:RTL:1609208.1609591},
and sensitivity
analysis~\cite{SensitivityDM,donze2010breach,DRJqest13,DuggiralaMV:2015c2e2}. Some
of the simulation based techniques only work for systems with linear
dynamics~\cite{Sim2Veri,ILABS}.  Recent work on the APEX tool~\cite{o2016apex} for verifying trajectory planning and tracking in autonomous vehicles is related our approach in that it targets the same application domain.

\section{Modeling/semantic framework}
\label{sec:prelims}

We introduce a powertrain control system from~\cite{jin2014powertrain}
as a running example to illustrate the elements of our hybrid system
modeling framework.

\subsection{Powertrain  control system}
\label{ex:powertrain}

This system ($\auto{Powertrn}$) models a highly nonlinear engine control system. The relevant state variables of the model are intake manifold pressure ($p$), air-fuel ratio ($\lambda$), estimated manifold pressure ($pe$) and intergrator state ($i$).
The overall system can be in one of four modes
$\Lmode{startup}$,
$\Lmode{normal}$,
$\Lmode{powerup}$, 
$\Lmode{sensorfail}$. 
A \Simulink\ diagram describes the continuous evolution of the above variables. 
In this paper, we mainly work on the {\em Hybrid I/O Automaton Model}
in the suite of powertrain control models. The \Simulink\ model
consists of continuous variables describing the dynamics of the
powertrain plant and sample-and-hold variables as the controller.  One
of the key requirements to verify is that the engine maintains the
air-fuel ratio within a desired range in different modes for a given
set of driver behaviors. This requirement has implications on fuel
economy and emissions.
For testing purposes, the control system designers work with sets of driver profiles that essentially define families of switching signals across the different modes. 
Previous verification results on this problem have been reported in~\cite{DFMV:CAV2015, FDMV:ARCH2015} on a simplified version of the powertrain control model.

\vspace{-10pt}
\subsection{Transition graphs}
\label{ssec:modes}
We will use  $\L$ to denote a finite set of {\em modes\/} or locations of the system under consideration. The discrete behavior or mode transitions are specified by what we call a transition graph over $\L$.
\begin{definition}
	 A {\em transition graph\/} is a labeled, directed acyclic
         graph $G = \langle \L, \V, \E, \vertlab, \edgelab \rangle$,
         where
	\begin{inparaenum}[(a)]
	\item $\L$ is the set of vertex labels also called the set of
          {\em modes\/},
	\item $\V$ the set of vertices, 
	\item $\E\subseteq \V \times \V$ is the set of edges, 
	\item $\vertlab: \V \rightarrow \L$ is a vertex labeling
          function that labels each vertex with a mode, and
	\item $\edgelab: \E\rightarrow \nnreals \times \nnreals$ is an
          edge labeling function that labels each edge with a
          nonempty, closed, bounded interval defined by pair of
          non-negative reals.
	\end{inparaenum}
\end{definition}
Since $G$ is a DAG, there is a nonempty subset $\V_{\sf init}
\subseteq \V$ of vertices with no incoming edges and a nonempty subset
$\V_{\sf term} \subseteq \V$ of vertices with no outgoing edges.  We
define the set of initial locations of $G$ as $\L_{\sf init} = \{ \ell
\ | \ \exists \ v \in \V_{\sf init}, \vertlab(v) = \ell \}$.
A (maximal) {\em path\/} of the graph $G$ is a sequence $\pi = v_1,
t_1, v_2, t_2, \ldots, v_k$ such that,
\begin{inparaenum}[(a)]
\item $v_1 \in \V_{\sf init}$,
\item $v_k \in \V_{\sf term}$, and 
\item for each $(v_i,t_i,v_{i+1})$ subsequence, there exists $(v_i, v_{i+1}) \in \E$, and $t_i \in \edgelab((v_i,v_{i+1}))$.
\end{inparaenum}
$\paths{G}$ is the set of all possible paths of $G$.  For a given path
$\pi = v_1, t_1, v_2, t_2, \ldots,$ $v_k$ its {\em trace\/}, denoted by
$\vertlab(\pi)$, is the sequence $\vertlab(v_1), t_1, \vertlab(v_2),
t_2, \ldots,$ $\vertlab(v_k)$.  Since $G$ is a DAG, a trace of $G$ can
visit the same mode finitely many times.  $\traces{G}$ is the set of
all traces of $G$.
%
%
%

An example transition graph for the $\auto{Powertrain}$ system of Section~\ref{ex:powertrain} is shown in Figure~\ref{fig:power_graph}. The set of vertices $\V = \{0,\ldots, 4\}$ and the $\vertlab$'s and $\edgelab$'s appear adjacent to the vertices and edges. 
\begin{figure}[h]
	\includegraphics[scale=0.4]{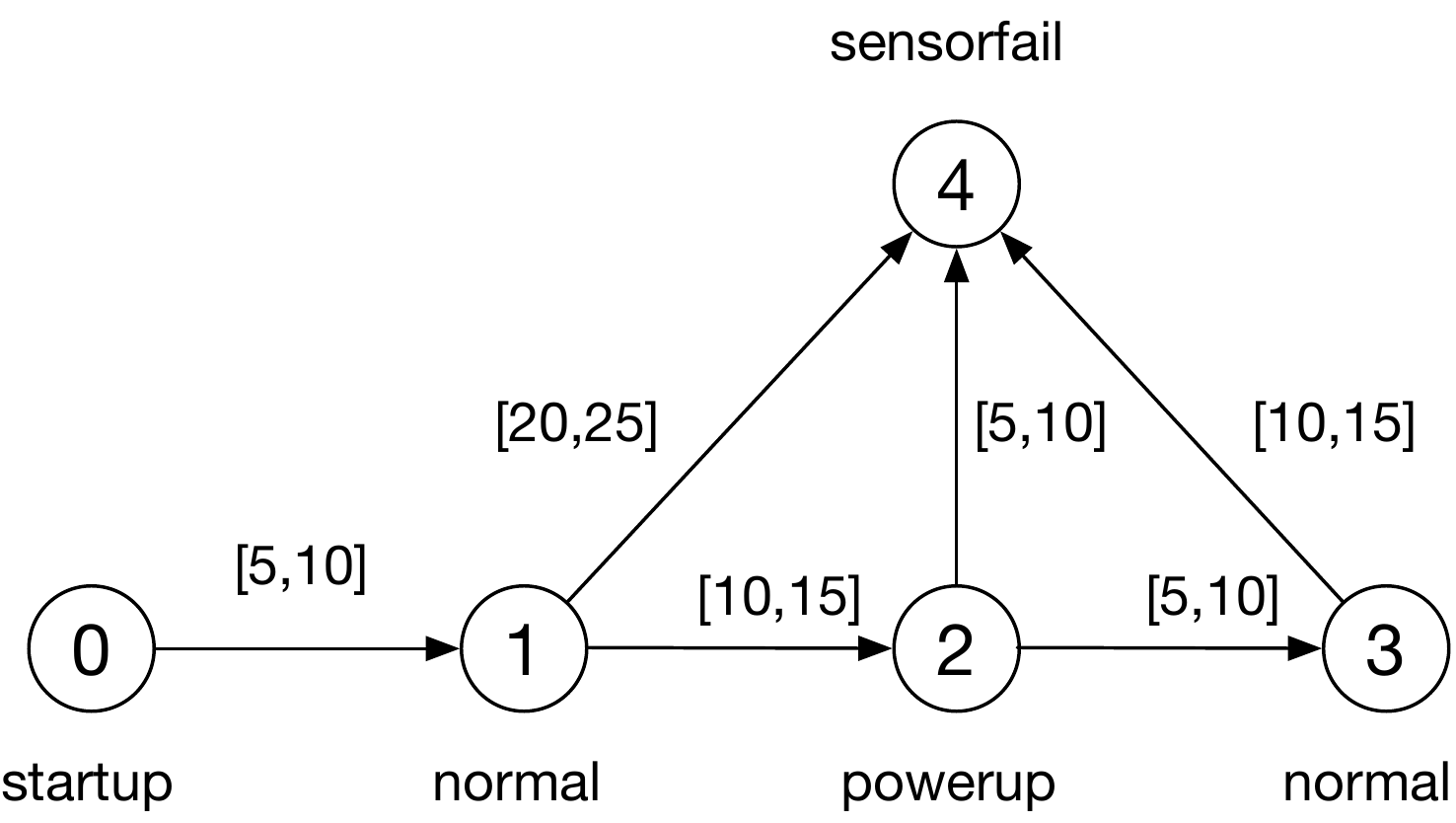}
	\vspace{-10pt}
	\caption{\small A sample transition graph for $\auto{Powertrain}$ system.}
	\label{fig:power_graph}	
\end{figure}
\vspace{-8pt}
\subsubsection{Trace containment}
We will develop reasoning techniques based on reachability,
abstraction, composition, and substitutivity. To this end, we will
need to establish containment relations between the behaviors of
systems. Here we define containment of transition graph traces.
Consider transition graphs $G_1, G_2,$ with modes $\L_1,\L_2,$ and a
mode map $\modemap: \L_1 \rightarrow \L_2$.  For a trace $\sigma =
\ell_1, t_1, \ell_2, t_2, \ldots, \ell_k \in \traces{G_1}$,
simplifying notation, we denote by $\modemap(\sigma)$ the sequence
$\modemap(\ell_1), t_1, \modemap(\ell_2), t_2,$ $\ldots,
\modemap(\ell_k)$.  We write $G_1 \preceq_{\modemap} G_2$ iff for
every trace $\sigma \in \traces{G_1}$, there is a trace $\sigma' \in
\traces{G_2}$ such that $\modemap(\sigma)$ is a prefix of $\sigma'$.

\begin{definition}
Given graphs $G_1, G_2$ and a mode map $\modemap: \L_1 \rightarrow
\L_2$, a relation $R \subseteq \V_1 \times \V_2$ is a {\em forward
  simulation relation from $G_1$ to $G_2$\/} iff
\begin{enumerate}[(a)]
\itemsep0em 
	\item for each $v \in \V_{1 {\sf init}}$, there is $u \in
          \V_{2 {\sf init}}$ such that $(v,u) \in R$,
        \item for every $(v,u) \in R$,  $\modemap(\vertlab_1(v)) =
          \vertlab_2(u)$, and
	\item for every $(v,v') \in \E_1$ and $(v,u)\in R$, there
          exists a finite set $u_1, \ldots, u_k$ such that:
	\begin{inparaenum}[(i)]
		\item for each $u_j$, $(v,u_j) \in R$, and
		\item $\edgelab_1((v,v')) \subseteq \cup_j
                  \edgelab_2((u,u_j))$.
	\end{inparaenum}
\end{enumerate}
\end{definition}

\begin{proposition}
	\label{prop:graphsim}
	If there exists a forward simulation relation from $G_1$ to
        $G_2$ with $\modemap$ then $G_1 \preceq_{\modemap} G_2$.
\end{proposition}

\vspace{-10pt}
\subsubsection{Sequential composition of graphs}

We will find it convenient to define the \emph{sequential composition}
of two transition graphs. Intuitively, the traces of the composition
of $G_1$ and $G_2$ will be those that can be obtained by concatenating
a trace of $G_1$ with a trace of $G_2$. To keep the definitions and
notations simple, we will assume (when taking sequential compositions)
$|\V_{\sf init}| = |\V_{\sf term}| = 1$; this is true of the examples
we analyze. It is easy to generalize to the case when this does not
hold. Under this assumption, the unique vertex in $\V_{\sf init}$ will
be denoted as $v_{\sf init}$ and the unique vertex in $\V_{\sf term}$
will be denoted as $v_{\sf term}$.

\begin{definition}
\label{def:seq-comp}
Given graphs $G_1 = \langle \L, \V_1, \E_1, \vertlab_1, \edgelab_1
\rangle$ and $G_2 = \langle \L, \V_2, \E_2,$ $\vertlab_2, \edgelab_2
\rangle$ such that $\vertlab_1(v_{1 {\sf term}}) = \vertlab_2(v_{2
  {\sf init}})$, the \emph{sequential composition} of $G_1$ and $G_2$
is the graph $G_1\seqcomp G_2 = \langle \L, \V, \E, \vertlab, \edgelab
\rangle$ where
\begin{enumerate}[(a)]
\itemsep0em 
\item $\V = (\V_1 \cup \V_2) \setminus \{v_{2 {\sf init}})\}$,
\item $\E = \E_1 \cup \{(v_{1 {\sf term}},u)\: |\: (v_{2 {\sf
    init}},u) \in \E_2\} \cup \{(v,u) \in \E_2\: |\: v \neq v_{2 {\sf
    init}}\}$,
\item $\vertlab(v) = \vertlab_1(v)$ if $v \in \V_1$ and $\vertlab(v) =
  \vertlab_2(v)$ if $v \in \V_2$,
\item For edge $(v,u) \in \E$, 
$\edgelab((v,u)) $ equals \begin{inparaenum}[(i)]
                   \item $\edgelab_1((v,u))$, if $u \in \V_1$, \\
                   \item $\edgelab_2((v_{2 {\sf init}},u))$,  if $v = v_{1 {\sf term}}$,
                   \item $\edgelab_2((v,u)), otherwise$.
                          \end{inparaenum}
\end{enumerate}
\end{definition}

Given our definition of trace containment between graphs, we can prove
a very simple property about sequential composition.

\begin{proposition}
\label{prop:seqcomp-containment}
Let $G_1$ and $G_2$ be two graphs with modes $\L$ that can be
sequential composed. Then $G_1 \preceq_{{\sf id}} G_1\seqcomp G_2$,
where ${\sf id}$ is the identity map on $\L$.
\end{proposition}

The proposition follows from the fact that every path of $G_1$ is a
prefix of a path of $G_1\seqcomp G_2$. Later in Section \ref{sec: sub_exp} we see examples of sequential composition.


%

\subsection{Trajectories}

The evolution of the system's continuous state variables is formally
described by continuous functions of time called {\em trajectories\/}.
Let $n$ be the number of continuous variables in the underlying hybrid
model. A {\em trajectory\/} for an $n$-dimensional system is a
continuous function of the form $\tau: [0,T] \rightarrow \reals^n$,
where $T \geq 0$. The interval $[0,T]$ is called the {\em domain\/} of
$\tau$ and is denoted by $\tau.\dom$. The first state $\tau(0)$ is
denoted by $\tau.\fstate$, last state $\tau.\lstate = \tau(T)$ and
$\tau.\ltime = T$.
For a hybrid system with $\L$ modes, each trajectory is labeled by a
mode in $\L$.  A {\em trajectory labeled by $\L$\/} is a pair $\langle
\tau, \ell \rangle$ where $\tau$ is a trajectory and $\ell \in \L$.
 
A {\em $T_1$-prefix\/} of $\langle \tau, \ell\rangle$, for any $T_1
\in \tau.\dom$, is the labeled-trajectory $\langle \tau_1,
\ell\rangle$ with $\tau_1:[0,T_1] \rightarrow \reals^n$, such that for
all $t \in [0, T_1]$, $\tau_1(t) = \tau(t)$.
A set of labeled-trajectories $\TL$ is prefix-closed if for any
$\langle \tau,\ell \rangle \in \TL$, any of its prefixes are also in
$\TL$.  A set $\TL$ is {\em deterministic\/} if for any pair $\langle
\tau_1, \ell_1 \rangle, \langle \tau_2,\ell_2 \rangle \in \TL$, if
$\tau_1.\fstate = \tau_2.\fstate$ and $\ell_1 = \ell_2$ then one is a
prefix of the other.
A deterministic, prefix-closed set of labeled trajectories $\TL$
describes the behavior of the continuous variables in modes $\L$.
We denote by $\TL_{\sf init,\ell}$ $=\{ \tau.\fstate \ | \ \langle
\tau,\ell \rangle \in \TL \}$, the set of initial states of
trajectories in mode $\ell$. Without loss generality we assume that
$\TL_{{\sf init},\ell}$ is a connected, compact subset of $\reals^n$.
We assume that trajectories are defined for unbounded time, that is,
for each $\ell \in \L, T >0$, and $x \in \TL_{{\sf init},\ell}$, there
exists a $\langle \tau, \ell \rangle \in \TL$, with $\tau.\fstate = x$
and $\tau.\ltime = T$.  

%
In control theory and hybrid systems literature, the trajectories are
assumed to be generated from models like ordinary differential
equations (ODEs) and differential algebraic equations (DAEs). Here,
we avoid an over-reliance on the models generating trajectories and
closed-form expressions. Instead, {\toolname} works with sampled data
of $\tau(\cdot)$ generated from simulations or tests.
	
\begin{definition}
	\label{def:sims}
A {\em simulator\/} for a (deterministic and prefix-closed) set $\TL$
of trajectories labeled by $\L$ is a function (or a program)
$\simulator$ that takes as input a mode label $\ell \in \L$, an
initial state $x_0 \in \TL_{\sf init,\ell}$, and a finite
sequence of time points $t_1, \ldots, t_k$, and returns a sequence of
states $\simulator(x_0,\ell,t_1), \ldots, \simulator(x_0,\ell, t_k)$
such that there exists $\langle\tau,\ell\rangle \in \TL$ with
$\tau.\fstate = x_0$ and for each $i\in \{1,\ldots, k\}$,
$\simulator(x_0,\ell,t_i) = \tau(t_i)$.
\end{definition}

The trajectories of the $\auto{Powertrn}$ system are described by a \Simulink\ diagram.
The diagram has several switch blocks and input signals  that can be set appropriately to generate simulation data using the \Simulink\ ODE solver.

For simplicity, we assume that the simulations are perfect (as in the
last equality of Definition~\ref{def:sims}). Formal guarantees of
soundness of {\toolname} are not compromised if we use \emph{validated
  simulations} instead.
	



\vspace{-10pt}
\paragraph{Trajectory containment}
Consider sets of trajectories, $\TL_1$ labeled by $\L_1$ and
$\TL_2$ labeled by $\L_2$, and a mode map $\modemap: \L_1 \rightarrow
\L_2$.
For a labeled trajectory $\langle \tau,\ell \rangle \in \TL_1$, denote by $\modemap(\langle \tau,\ell \rangle)$ the labeled-trajectory
$\langle \tau,\modemap(\ell) \rangle$.  Write $\TL_1
\preceq_{\modemap} \TL_2$ iff for every labeled trajectory $\langle
\tau,\ell \rangle \in \TL_1$, $\modemap(\langle \tau,\ell \rangle) \in
\TL_2$.

\subsection{Hybrid systems}
\begin{definition}	
An $n$-dimensional  {\em hybrid system\/} $\H$ is a 4-tuple $\langle \L, \Theta, G, \TL \rangle$, where
\begin{inparaenum}[(a)]
\item $\L$ is a finite set of modes, 
\item $\Theta \subseteq \reals^n$ is a compact set of initial states,
\item $G = \langle \L, \V, \E, \edgelab \rangle$ is a transition graph
  with set of modes $\L$, and
\item $\TL$ is a set of deterministic, prefix-closed trajectories
  labeled by $\L$.
\end{inparaenum}
\end{definition}

A {\em state\/} of the hybrid system $\H$ is a point in $\reals^n \times \L$. 
The set of initial states is $\Theta \times \L_{\sf init}$.  Semantics
of $\H$ is given in terms of executions which are sequences of
trajectories consistent with the modes defined by the transition
graph.  An {\em execution\/} of $\H$ is a sequence of labeled
trajectories $\alpha = \langle \tau_1, \ell_1\rangle\ldots, \langle
\tau_{k-1}, \ell_{k-1}\rangle, \ell_k$ in $\TL$, such that
\begin{inparaenum}[(a)]
	\item $\tau_1.\fstate \in \Theta$ and  $\ell_1 \in \L_{\sf init}$,
	\item the sequence $\pathof{\alpha}$ defined as $\ell_1, \tau_1.\ltime, \ell_2, \ldots \ell_k$ is in $\traces{G}$, and
	\item for each consecutive trajectory, $\tau_{i+1}.\fstate =
          \tau_i.\lstate$.
\end{inparaenum}
The set of all executions of $\H$ is denoted by $\execs{\H}$.  The
first and last states of an execution $\alpha = \langle \tau_1,
\ell_1\rangle\ldots, \langle \tau_{k-1}, \ell_{k-1}\rangle, \ell_k$
are $\alpha.\fstate = \tau_1.\fstate$, $\alpha.\lstate =
\tau_{k-1}.\lstate$, and $\alpha.\fmode = \ell_1$ $\alpha.\lmode =
\ell_k$.  A state $\langle x, \ell \rangle$ is {\em reachable\/} at
time $t$ and vertex $v$ (of graph $G$) if there exists an execution
$\alpha = \langle \tau_1, \ell_1\rangle\ldots, \langle \tau_{k-1},
\ell_{k-1}\rangle, \ell_k \in \execs{\H}$, a path $\pi = v_1, t_1,
\ldots v_k$ in $\paths{G}$, $i \in \{1,\ldots k\}$, and $t' \in
\tau_i.\dom$ such that $\vertlab(\pi) = \pathof{\alpha}$, $v = v_i$,
$\ell = \ell_i$, $x = \tau_i(t')$, and $t = t' + \sum_{j=1}^{i-1} t_j$.
The set of reachable states, reach tube, and states reachable at a
vertex $v$ are defined as follows.

\begin{itemize}[\null]
\itemsep0pt 
\item $\reachtube{\H} = \{\langle x,\ell,t \rangle\: |\: \mbox{for some }
v,\ \langle x, \ell \rangle \mbox{ is reachable at time $t$ and vertex
  $v$}\}$
\item $\reach{\H} = \{\langle x,\ell \rangle\: |\: \mbox{for some }
v,t,\ \langle x, \ell \rangle \mbox{ is reachable at time $t$ and
  vertex $v$}\}$
\item $\reach{\H}^v = \{\langle x,\ell \rangle\: |\: \mbox{for some }
t,\ \langle x, \ell \rangle \mbox{ is reachable at time $t$ and
  vertex $v$}\}$
\end{itemize}

Given
a set of (unsafe) states $\U \subseteq \reals^n \times \L$, the {\em
  bounded safety verification problem\/} is to decide whether
$\reach{\H} \cap \U = \emptyset$.  In Section~\ref{sec:reachalgo} we
will present \toolname's algorithm for solving this decision problem.

\begin{remark}
Defining paths in a graph $G$ to be maximal (i.e., end in a vertex in
$\V_{{\sf term}}$) coupled with the definition above for executions in
$\H$, ensures that for a vertex $v$ with outgoing edges in $G$, the
execution must leave the mode $\vertlab(v)$ within time bounded by the
largest time in the labels of outgoing edges from $v$.
\end{remark}

An instance of the bounded safety verification problem is defined by
(a) the hybrid system for the $\auto{Powertrn}$ which itself is
defined by the transition graph of Figure~\ref{fig:power_graph} and
the trajectories defined by the \Simulink\ model, and (b) the unsafe
set ($\U_p$): in \Lmode{powerup} mode, $t>4\wedge \lambda \notin
[12.4,12.6]$, in \Lmode{normal} mode, $t>4 \wedge \lambda \notin
[14.6,14.8]$.

Containment between graphs and trajectories can be leveraged to
conclude the containment of the set of reachable states of two hybrid
systems.
\begin{proposition}
\label{prop:hybrid-contain}
Consider a pair of hybrid systems $\H_i = \langle \L_i, \Theta_i, G_i,
\TL_i \rangle$, $i \in \{1,2\}$ and mode map $\modemap: \L_1 \to
\L_2$. If $\Theta_1 \subseteq \Theta_2$, $G_1 \preceq_{\modemap} G_2$,
and $\TL_1 \preceq_{\modemap} \TL_2$, then $\reach{\H_1} \subseteq
\reach{\H_2}$.
\end{proposition}



\subsection{ADAS and autonomous vehicle benchmarks}
\label{ssec:adas}
This is a suite of benchmarks we have created representing various common scenarios used for testing ADAS and Autonomous driving control systems. 
The hybrid system for a scenario is  constructed by putting together several individual vehicles.
The higher-level decisions (paths) followed by the vehicles are captured by transition graphs
while the detailed dynamics of each vehicle comes from a black-box 
\Simulink \ simulator from \Mathworks~\cite{simulinkcar}.

Each vehicle has several continuous variables including the $x,
y$-coordinates of the vehicle on the road, its velocity, heading, and
steering angle. The vehicle can be controlled by two input signals,
namely the throttle (acceleration or brake) and the steering speed.
By choosing appropriate values for these input signals, we have
defined the following modes for each vehicle ---
\Lmode{cruise}: move forward at constant speed, 
\Lmode{speedup}: constant acceleration,
\Lmode{brake}: constant (slow) deceleration,
\Lmode{em\_brake}: constant (hard) deceleration.
In addition, we have designed lane switching modes \Lmode{ch\_left} and \Lmode{ch\_right} in which the acceleration and steering are controlled in such a manner that the vehicle switches to its left (resp. right) lane in a certain amount of time. 

For each vehicle, we mainly analyze four variables: absolute position
($sx$) and velocity ($vx$) orthogonal to the road direction
($x$-axis), and absolute position ($sy$) and velocity ($vy$) along the
road direction ($y$-axis). The throttle and steering are captured using
the four variables. We will use subscripts to distinguish between
different vehicles.  The following scenarios are constructed by
defining appropriate sets of initial states and transitions graphs
labeled by the modes of two or more vehicles. In all of these
scenarios a primary safety requirement is that the vehicles maintain
safe separation. See Appendix~\ref{app:adas} for more details on
initial states and transition graphs of each scenario.
\begin{description} 
\itemsep0em 
\item[$\auto{Merge}$:] 
Vehicle A in the left lane is behind   vehicle B  in the right lane. A switches through modes  \Lmode{cruise}, \Lmode{speedup}, \Lmode{ch\_right}, and \Lmode{cruise} over  specified intervals to merge behind B.
Variants of this scenario involve $B$ also switching to \Lmode{speedup} or \Lmode{brake}.
%
\item[$\auto{AutoPassing}$:]
Vehicle A starts behind  B in the same lane, and  goes through a sequence of modes to  overtake B. If B switches to \Lmode{speedup} before A enters \Lmode{speedup} then A aborts and changes back to right lane. 
\item[$\auto{Merge3}$:]
Same as $\auto{AutoPassing}$ with a third car C always ahead of $B$.
\item[$\auto{AEB}$:]
Vehicle A cruises behind B and B stops.
A transits from \Lmode{cruise} to \Lmode{em\_brake} possibly over several different time intervals as governed by different sensors and reaction times.

\end{description}





\section{Invariant verification}
\label{sec:reachalgo}

A subproblem for invariant verification is to
compute $\reachtube{\H}$, or more specifically, the reachtubes for the set of trajectories $\TL$ in a given mode, up to a time bound.
This is a difficult problem, even when $\TL$ is generated by white-box models.
The algorithms in~\cite{donze2010breach,DMV:EMSOFT2013,FanMitra:2015df} approximate reachtubes using simulations and sensitivity analysis of ODE models generating $\TL$.
Here, we begin with a probabilistic method for estimating sensitivity from black-box simulators.

\subsection{Discrepancy functions}
\label{sec:disc}

Sensitivity of trajectories is formalized by the notion of discrepancy
functions \cite{DMV:EMSOFT2013}.  For a set 
$\TL$, a {\em discrepancy function\/} is a uniformly continuous
function $\beta: \reals^n \times \reals^n \times \nnreals \rightarrow
\nnreals$, such that for any pair of identically labeled trajectories
$\langle \tau_1,\ell \rangle, \langle \tau_2, \ell \rangle \in \TL$,
and any $t \in \tau_1.\dom \cap \tau_2.\dom$:
\begin{inparaenum}[(a)]
	\item $\beta$ upper-bounds the distance between the trajectories, i.e., 
	\begin{align}
	|\tau_1(t) - \tau_2(t)|  \leq \beta(\tau_1.\fstate,\tau_2.\fstate,t), \label{eq:discrepancy}
	\end{align} and
	\item $\beta$ converges to $0$ as the initial states converge, i.e., for any trajectory $\tau$ and $t \in
          \tau.\dom$, if a sequence of trajectories $\tau_1,\ldots,
          \tau_k, \ldots$ has $\tau_k.\fstate \rightarrow
          \tau.\fstate$, then $\beta(\tau_k.\fstate,$
          $\tau.\fstate,t)$ $\rightarrow 0$.
\end{inparaenum}
In~\cite{DMV:EMSOFT2013} it is shown how given a $\beta$, condition~(a) can used to over-approximate reachtubes  from simulations, 
and condition~(b) can be used to make these approximations arbitrarily precise. 
Techniques for computing $\beta$ from ODE models are developed in~\cite{FanMitra:2015df,FanM:EMSOFT2016,HFMMK:CAV2014}, but these are not applicable here in absence of such models. Instead we present a simple method for discovering discrepancy functions that only uses simulations. Our method is based on classical results on PAC learning linear separators~\cite{KearnsVazirani}. We recall these before applying them to find discrepancy functions.
\vspace{-10pt}
\subsubsection{Learning linear separators.}
\label{Sec: discrepancy}
For $\Gamma \subseteq \reals\times\reals$, a \emph{linear separator}
is a pair $(a,b) \in \reals^2$ such that
\begin{align}
\forall (x,y) \in \Gamma.\ x \leq ay + b.
\label{eq:separator}
\end{align}
Let us fix a subset $\Gamma$ that has a (unknown) linear separator
$(a_*,b_*)$. Our goal is to discover some $(a,b)$ that is a linear
seprator for $\Gamma$ by sampling points in $\Gamma$~\footnote{We
  prefer to present the learning question in this form as opposed to
  one where we learn a Boolean concept because it is closer to the
  task at hand.}. The assumption is that elements of $\Gamma$ can be
drawn according to some (unknown) distribution ${\cal D}$. With
respect to ${\cal D}$, the \emph{error} of a pair $(a,b)$ from
satisfying Equation~\ref{eq:separator}, is defined to be
$\mathsf{err}_{{\cal D}}(a,b) = {\cal D}(\{(x,y) \in \Gamma\: |\: x >
ay+b\})$
where ${\cal D}(X)$ is the measure of set $X$ under distribution
${\cal D}$. Thus, the error is the measure of points (w.r.t. ${\cal
  D}$) that $(a,b)$ is not a linear separator for. There is a very
simple (probabilistic) algorithm that finds a pair $(a,b)$ that is a
linear separator for a large fraction of points in $\Gamma$, as
follows.
\begin{enumerate}
\itemsep0em 
\item\label{alg:linsep1} Draw $k$ pairs $(x_1,y_1), \ldots (x_k,y_k)$
  from $\Gamma$ according to ${\cal D}$; the value of $k$ will be
  fixed later.
\item\label{alg:linsep2} Find $(a,b) \in \reals^2$ such that $x_i
  \leq ay_i + b$ for all $i \in \{1,\ldots k\}$.
\end{enumerate}
Step~\ref{alg:linsep2} involves checking feasibility of a linear
program, and so can be done efficiently. This algorithm, with high
probability, finds a linear separator for a large fraction of points.
\begin{proposition}
\proplabel{linear-sep-learn} Let $\epsilon, \delta \in \plreals$. If
$k \geq \frac{1}{\epsilon}\ln\frac{1}{\delta}$ then, with probability
$\geq 1-\delta$, the above algorithm finds $(a,b)$ such that
$\mathsf{err}_{{\cal D}}(a,b) < \epsilon$.
\end{proposition}
%
\begin{proof}
The result follows from the PAC-learnability of concepts with low
VC-dimension~\cite{KearnsVazirani}. However, since the proof is very
simple in this case, we reproduce it here for completeness. Let $k$ be
as in the statement of the proposition, and suppose the pair $(a,b)$
identified by the algorithm has error $> \epsilon$. We will bound the
probability of this happening.

Let $B = \{(x,y)\: |\: x > ay+b\}$. We know that ${\cal D}(B) >
\epsilon$. The algorithm chose $(a,b)$ only because no element from
$B$ was sampled in Step~\ref{alg:linsep1}. The probability that this
happens is $\leq (1-\epsilon)^k$. Observing that $(1-s) \leq e^{-s}$
for any $s$, we get $(1-\epsilon)^k \leq e^{-\epsilon k} \leq e^{-\ln
  \frac{1}{\delta}} = \delta$.  This gives us the desired result.
\end{proof}

\subsubsection{Learning discrepancy functions}
Discrepancy functions will be computed from simulation data
independently for each mode. Let us fix a mode $\ell \in \L$, and a
domain $[0,T]$ for each trajectory. The discrepancy functions that we
will learn from simulation data, will be one of two different forms,
and we discuss how these are obtained.

\vspace{-10pt}
\paragraph{Global exponential discrepancy (GED)} is a function of the form
\[
\beta(x_1,x_2,t) = |x_1 - x_2| Ke^{\gamma t}.
\]
Here $K$ and $\gamma$ are constants. Thus, for any pair of
trajectories $\tau_1$ and $\tau_2$ (for mode $\ell$), we have
\[
\forall t \in [0,T].\ |\tau_1(t) - \tau_2(t)| \leq |\tau_1.\fstate -
\tau_2.\fstate| Ke^{\gamma t}.
\]
Taking logs on both sides and rearranging terms, we have
\[
\forall t.\ \ln \frac{|\tau_1(t) - \tau_2(t)|}{|\tau_1.\fstate -
\tau_2.\fstate|} \leq \gamma t + \ln K.
\]
It is easy to see that a global exponential discrepancy is nothing but
a linear separator for the set $\Gamma$ consisting of pairs $(\ln
\frac{|\tau_1(t) = \tau_2(t)|}{|\tau_1.\fstate - \tau_2.\fstate|}, t)$
for all pairs of trajectories $\tau_1,\tau_2$ and time $t$. Using the
sampling based algorithm described before, we could construct a GED for a mode $\ell \in \L$, where sampling from
$\Gamma$ reduces to using the simulator to generate traces from
different states in $\TL_{\sf init, \ell}$. \propref{linear-sep-learn}
guarantees the correctness, with high probability, for any separator
discovered by the algorithm.  However, for our reachability algorithm
to not be too conservative, we need $K$ and $\gamma$ to be
small. Thus, when solving the linear program in Step~\ref{alg:linsep2}
of the algorithm, we search for a solution minimizing  $\gamma T +
\ln K$.

\vspace{-10pt}
\paragraph{Piece-wise exponential discrepancy (PED).} 
The second form of discrepancy functions we consider, depends upon
dividing up the time domain $[0,T]$ into smaller intervals, and
finding a global exponential discrepancy for each interval. Let $0 =
t_0,t_1,\ldots t_N = T$ be an increasing sequence of time points. Let
$K, \gamma_1, \gamma_2, \ldots \gamma_N$ be such that for every pair
of trajectories $\tau_1,\tau_2$ (of mode $\ell$), for every $i \in
\{1,\ldots, N\}$, and $t \in [t_{i-1},t_i]$, $|\tau_1(t) = \tau_2(t)|
\leq |\tau_1(t_{i-1}) - \tau_2(t_{i-1})| Ke^{\gamma_i t}$.  Under such
circumstances, the discrepancy function itself can be seen to be given
as
\[
\beta(x_1,x_2,t) = |x_1 - x_2| Ke^{\sum_{j=1}^{i-1}\gamma_j(t_j -
  t_{j-1}) + \gamma_i (t-t_{i-1})} \qquad \mbox{for } t \in [t_{i-1},t_i].
\]
If the time points $0 = t_0,t_1,\ldots t_N = T$ are fixed, then the
constants $K, \gamma_1, \gamma_2, \ldots $ $\gamma_N$ can be discovered
using the learning approach described for GED; here, to discover $\gamma_i$, we take $\Gamma_i$ to be the pairs obtained by restricting the trajectories to be between times
$t_{i-1}$ and $t_i$. The sequence of time points $t_i$ are also
dynamically constructed by our algorithm based on the following
approach. Our experience suggests that a value for $\gamma$ that is
$\geq 2$ results in very conservative reach tube
computation. Therefore, the time points $t_i$ are constructed
inductively to be as large as possible, while ensuring that $\gamma_i <
2$.

\vspace{-10pt}
\subsubsection{Experiments on learning discrepancy}

We used the above algorithm to learn discrepancy functions for dozens
of modes with complex, nonlinear trajectories.  Our experiments
suggest that around 10-20 simulation traces are adequate for computing
both global and piece-wise discrepancy functions. For each mode we use
a set $S_{\sf train}$ of simulation traces that start from
independently drawn random initial states in $\TL_{\sf init,\ell}$ to
learn a discrepancy function.  Each trace may have $100$-$10000$ time
points, depending on the relevant time horizon and sample times.  Then
we draw another set $S_{\sf test}$ of $1000$ simulations traces for
validating the computed discrepancy. For every pair of trace in
$S_{\sf test}$ and for every time point, we check whether the computed
discrepancy satisfies Equation~\ref{eq:discrepancy}.  We 
observe that for $|S_{\sf train}| > 10$ the computed discrepancy
function is correct for $96\%$ of the points $S_{\sf test}$ in and for
$|S_{\sf train}| > 20$ it is correct for more than $99.9\%$, across all experiments.

\subsection{Verification algorithm}
\label{sec:verfication_algorithm}

In this section, we present algorithms to solve the bounded
verification problem for hybrid systems using learned exponential discrepancy functions.
We first introduce an algorithm $\GraphReach$ (Algorithm
\ref{alg:ComputeRT}) which takes as input a hybrid system $\H =
\langle \L, \Theta, G, \TL \rangle$ and returns a set of
reachtubes---one for each vertex of $G$---such that their
union over-approximates $\reachtube{\H}$.
 



$\GraphReach$ maintains two  data-structures:
 \begin{inparaenum}[(a)]
\item $RS$ accumulates pairs of the form $\langle RT, v\rangle$, where  $v \in \V$ and $RT$ is its corresponding reachtube;
\item $\VerInit$ accumulates pairs of the form $\langle S, v \rangle$,
  where $v \in \V$ and $S\subset \reals^n$ is the set of states  from which the reachtube in $v$ is to be computed.
\end{inparaenum} 
Each
$v$ could be in multiple such pairs in $RS$ and $\VerInit$.
Initially, $RS = \emptyset$  and
$\VerInit = \{\langle \Theta, v_{\sf init}\rangle\}$.

$\mathit{LearnDiscrepancy(S_{\sf init},d,\ell)}$ computes
the discrepancy function for mode $\ell$, from initial set $S_{\sf
  init}$ and upto time $d$ using the algorithm of Section \ref{sec:disc}. $\computeRT(S_{\sf init},d,\beta)$ first
generates finite simulation traces from  $S_{\sf init}$
and then bloats the traces to compute a reachtube using the discrepancy function $\beta$. This step is similar to the algorithm for dynamical systems  given in~\cite{DMV:EMSOFT2013}.

The $\GraphReach$ algorithm proceeds as follows: first, a topologically sorted array of the vertices of the DAG $G$ is computed in  $\mathit{Order}$ (\lnref{ln: init2}). 
The pointer $ptr$ iterates over the $\mathit{Order}$ and for each vertex $\mathit{curv}$ the following is computed. 
The variable $\mathit{dt}$  is set to the maximum transition time to other vertices from $\mathit{curv}$ (\lnref{ln: dwt}). 
For each possible initial set $S_{\sf init}$ corresponding to $\mathit{curv}$ in $\VerInit$, the algorithm computes a discrepancy function (\lnref{ln: disc}) and uses it to compute a reachtube from $S_{\sf init}$ up to time $\mathit{dt}$ (\lnref{ln: reachtube}). 
For each successor $\mathit{nextv}$ of $\mathit{curv}$, the restriction of the computed reachtube $RT$ to the corresponding transition time interval $\edgelab((\mathit{curv,nextv}))$ is set as  an initial set for $\mathit{nextv}$ (\lnsref{ln: forloopbegin}{ln: forloopend}).

\begin{algorithm}[h!]
\caption{$\GraphReach(\H)$ computes bounded time reachtubes for each vertex of the transition $G$ of hybrid system $\H$.}
\label{alg:ComputeRT}
\SetKwInOut{Input}{input}
\SetKwInOut{Initially}{initially}
{$RS \gets \emptyset; \VerInit \gets \{\langle \Theta, v_{\sf init} \rangle\}; \mathit{Order} \gets \mathit{TopSort}(G)$\;} \lnlabel{ln: init2}
\For {$ptr = 0: len(Order)-1$}
{	
	{$\mathit{curv} \gets \mathrm{Order}[ptr]$ \;} \lnlabel{ln: currentV}
	{$\ell \gets \vertlab(\mathit{curv})$\;} \lnlabel{ln: currentL}

	$\mathit{dt} \gets \textrm{max} \{t' \in \nnreals \:| \:\exists vs \in \V, (\mathit{curv}, vs) \in \E, (t,t') \gets \edgelab \left( (\mathit{curv}, vs) \right) \}$\; \lnlabel{ln: dwt}

	\For {$S_{\sf init} \in \{S~|~ \langle S,\mathit{curv} \rangle \in \VerInit\}$}
	{
		{$\beta \gets \mathit{LearnDiscrepancy}(S_{\sf init},\mathit{dt},\ell)$\;} \lnlabel{ln: disc}
		{$  RT \gets \computeRT(S_{\sf init},\mathit{dt},\beta)$\;} \lnlabel{ln: reachtube}
		{$RS \gets RS \cup \langle RT, \mathit{curv} \rangle $\;}
		\For {$\mathit{nextv} \in \mathit{curv}.succ$} 
		{
			$(t,t') \gets \edgelab \left( (\mathit{curv}, nextv) \right)$\; \lnlabel{ln: forloopbegin}
			$\VerInit \gets \VerInit \cup \langle  \mathit{Restr}(RT,(t,t')), nextv\rangle$\; \lnlabel{ln: forloopend}
		}
	}
}
\Return $RS$ \;
\end{algorithm}

The invariant verification algorithm $\mathit{VerifySafety}$ decides
safety of $\H$ with respect to a given unsafe set $\U$ and uses
$\GraphReach$. The detailed pseudocode appears in
Appendix~\ref{appendix:safetyveri}.  This algorithm proceeds in a way
similar to the simulation-based verification algorithms for dynamical
and hybrid systems~\cite{DMV:EMSOFT2013,fan2016automatic}.
Given initial set $\Theta$ and transition graph $G$ of $\H$, this algorithm partitions $\Theta$ into several subsets, and then for each subset $S$ it checks whether the computed  over-approximate reachtube $RS$ from $S$ intersects with $\U$:
 \begin{inparaenum}[(a)]
 \item If $RS$  is disjoint, the system is safe starting from $S$; 
 \item if certain part of a reachtube $RT$ is contained in $\U$, the system is declared as  unsafe and $RT$ with the the corresponding path of the graph are returned as counter-example witnesses; 
 \item if neither of the above conditions hold, then the algorithm performs refinement to get a more precise over-approximation of $RS$.
 \end{inparaenum} 
Several refinement strategies are implemented in {\toolname} to
accomplish the last step.  Broadly, these strategies rely on splitting
the initial set $S$ into smaller sets (this gives tighter discrepancy
in the subsequent vertices) and splitting the edge labels of $G$ into
smaller intervals (this gives smaller initial sets in the vertices).
 
The above description focuses on invariant properties, but the
algorithm and our implementation in {\toolname} can verify a useful
class of temporal properties.  These are properties in which the time
constraints only refer to the time since the last mode transition.
For example, for the $\auto{Powertrn}$ benchmark the tool verifies
requirements like ``after $4$s in $\Lmode{normal}$ mode, the air-fuel
ratio should be contained in $[14.6,14.8]$ and after $4$s in
$\Lmode{powerup}$ it should be in $[12.4,12.6]$''.

\vspace{-10pt}
\paragraph{Correctness}
Given a correct discrepancy function for each mode, we can prove the
soundness and relative completeness of
Algorithm~\ref{alg:safetyveri}. This analysis closely follows the
proof of Theorem~19 and Theorem~21
in~\cite{duggirala2015dynamic}. Combining this with the probabilistic
correctness of the $\mathit{LearnDiscrepancy}$, we obtain the
following probabilistic soundness guarantee.

\begin{theorem}
	If the $\beta$'s returned by $\mathit{LearnDiscrepancy}$ 
	are always  discrepancy functions for corresponding modes, then $\mathit{VerifySafety}(\H,U)$ (Algorithm~\ref{alg:safetyveri}) is sound. That is, if it outputs ``SAFE'', then $\H$  is safe with respect to $\U$ and if it outputs ``UNSAFE'' then there exists an execution of $\H$ that enters $\U$.
\end{theorem}

\subsection{Experiments on safety verification}
\label{sec:reachexp}

\begin{figure}
	\centering
	\begin{subfigure}[b]{0.45\textwidth}
		\includegraphics[width=\textwidth,trim={1.0cm 0.8cm 1.5cm 0.8cm},clip]{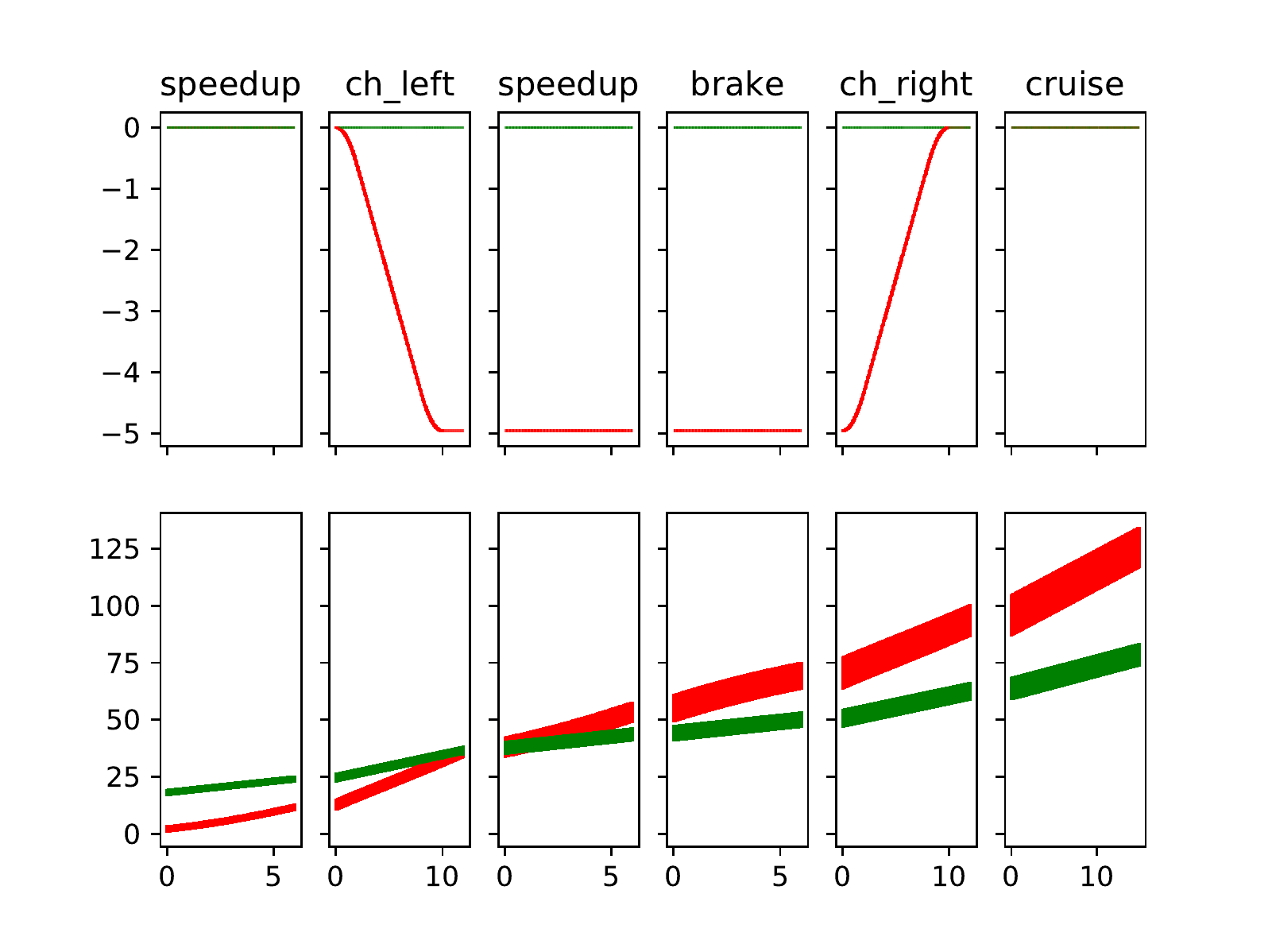}
		\caption{Safe reachtube. 
			}
		\label{fig:AutoPassingA}
	\end{subfigure}
	~ ~
	\begin{subfigure}[b]{0.45\textwidth}
		\includegraphics[width=\textwidth,trim={1.0cm 0.8cm 1.5cm 0.8cm},clip]{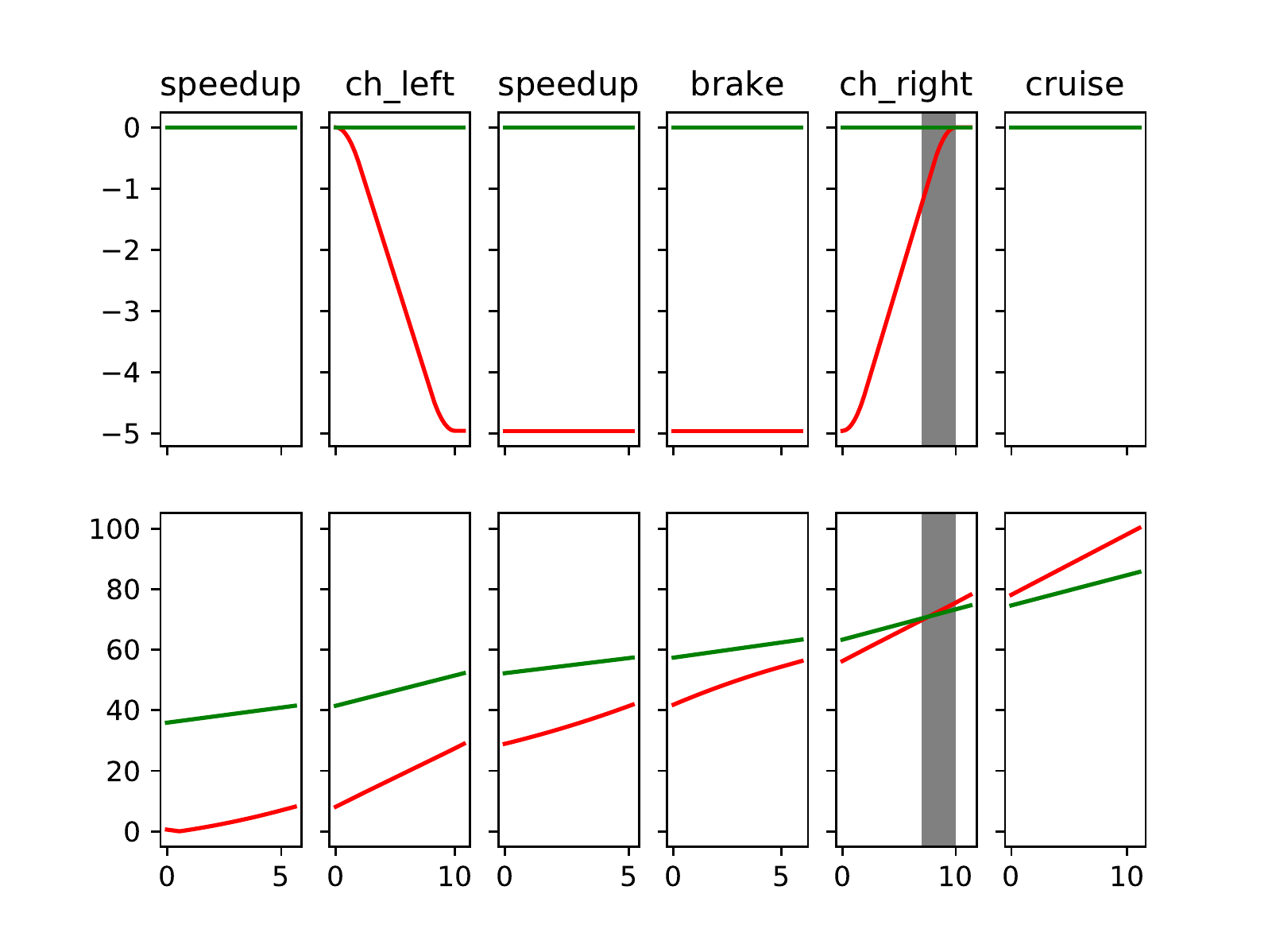}
		\caption{Unsafe execution. 
			}
		\label{fig:AutoPassingB}
	\end{subfigure}
\vspace{-8pt}
\caption{\small $\auto{AutoPassing}$  verification. Vehicle A's (red) modes are shown above each subplot. Vehicle B (green) is in \Lmode{cruise}. 
	Top: $sx_A,sx_B$. Bottom: $sy_A, sy_B$.
	 }
\label{fig:AutoPassing}
\end{figure}

The algorithms have been implemented in {\toolname} and
have been used to automatically verify the benchmarks from
Section~\ref{sec:prelims} and an Automatic Transmission System (Appendix~\ref{ssec:gear}). 
The transition graph, the initial set, and unsafe set are given in a text file. {\toolname} uses simulators for modes, and outputs either ``Safe'' of ``Unsafe''. Reachtubes or counter-examples computed during the analysis are also stored in text files.

The implementation is in Python using the MatLab's Python API for accessing the \Simulink ~simulators. Py-GLPK~\cite{python-glpk-package} is used to find the
parameters of discrepancy functions;
either global (GED) or piece-wise (PED) discrepancy can be selected by
the user.  Z3~\cite{de2008z3} is used for reachtube operations.
	At this stage, all the benchmarks we are working on heavily
  rely on {\Mathworks} {\Simulink}. We don't have a public
  {\Mathworks} license to release the tool, and it is complicated for
  the users to build a connection between {\toolname} and
  their own {\Simulink} models. We will release {\toolname} soon after
  we move the blackbox benchmarks to a different open source software.

Figure~\ref{fig:AutoPassing} shows example plots of computed safe
reachtubes and counter-examples for a simplified
$\auto{AutoPassing}$ in which vehicle B stays in the \Lmode{cruise} always.
As before, vehicle A goes through a sequence of  modes 
to overtake B. 
Initially, for both $i \in \{A,B\}$, $sx_i = vx_i = 0$ and $vy_i = 1$, i.e., both are cruising at constant speed at the center of the right lane; initial positions along the lane are $sy_A\in [0,2], sy_B \in
[15,17]$.  Figure~\ref{fig:AutoPassingA} shows the lateral positions
($sx_A$ in red and $sx_B$ in green, in the top subplot), and the
positions along the lane ($sy_A$ in red and $sy_B$ in green, in the
bottom plot).
 Vehicle A moves to left lane ($sx$ decreases) and then
back to the right, while B remains in the right lane, as A overtakes
B (bottom plot).  The unsafe set $(|sx_A-sx_B|<2 \ \&\ |sy_A-sy_B|<2)$ is proved to be disjoint from computed reachtube.  With a different
initial set, $sy_B \in [30,40]$, \toolname\ finds counter-example (Figure~\ref{fig:AutoPassingB}).


\begin{table}
\begin{center}
  \centering
\resizebox{\columnwidth}{!}{%
    \begin{tabular}{|c|rlcccr|}
    \hline
    Model &  TH &  Initial set & $\U$ & Ref & Safe & Runtime  \\
    \hline
    \makecell[c]{$\auto{Powertrn}$ \\ (5 vers, 6 edges)}   & 80    & \makecell[l]{$\lambda \in [14.6,14.8]$} & $\U_{p}$      & 2     & \cmark & 217.4s \\
\cline{1-7}
    {$\auto{AutoPassing}$}   & 50    & $sy_A \in [-1,1]$ $sy_B \in [14,16]$ & $\U_{c}$    & 4     & \cmark  & 208.4s \\ \cline{2-7}
     (12 vers, 13 edges) & 50    &  $sy_A \in [-1,1]$ $sy_B \in [4,6.5]$ & $\U_{c}$     & 5     & \xmark & 152.5s \\ \hline
    {$\auto{Merge}$ } &   50    & $sx_A \in [-5,5]$ $sy_B\in [-2,2]$ & $\U_{c}$      & 0     & \cmark  & 55.0s \\ \cline{2-7}
     (7 vers, 7 edges) &   50    &  $sx_A \in [-5,5]$  $sy_B \in [2,10]$ & $\U_{c}$      & -     & \xmark & 38.7s \\ \hline
    {$\auto{Merge3}$} &   50    &  \makecell[l]{$sy_A \in [-3,3]$ $sy_B \in [14, 23]$\\$sy_C \in [36, 45]$} & $\U_{c}$      & 4     & \cmark  & 197.6s \\ \cline{2-7}
     (6 vers, 5 edges) &   50    & \makecell[l]{$sy_A \in [-3,3]$ $sy_B \in [14,15]$ \\ $sy_C \in [16, 20]$} & $\U_{c}$     & -     & \xmark & 21.3s \\ \hline
   \makecell[c]{ {$\auto{ATS}$ }\\ (4 vers, 3 edges)} &  50    & Erpm $\in [900,1000]$& $\U_{t}$    & 2     & \cmark  & 109.2s \\ \cline{2-7}
    \hline
    \end{tabular}%
}
\vspace{-14pt}
  \caption{\small Safety verification results. Numbers below benchmark names: \# vertices and edges of $G$, TH: duration of shortest path in $G$,  Ref: \# refinements performed; Runtime: overall running time.}
\label{table:results}
  \end{center}
\end{table}%

Table~\ref{table:results} summarizes some of the verification results obtained using \toolname. $\auto{ATS}$ is an automatic
transmission control system (see Appendix~\ref{ssec:gear} for more details). These experiments were
performed on a laptop with Intel Core i7-6600U CPU and 16 GB RAM. 
%
The initial range of only the salient continuous variables are shown in the table.  The unsafe sets are discussed with the model description. For example $\U_c$ means  two vehicles are too close.
%
%
For all the benchmarks, the algorithm terminated in a few minutes which includes the time to simulate, learn discrepancy, generate reachtubes, check the safety of the reachtube, over all refinements.

%
 For the results presented in Table~\ref{table:results}, we used GED. The reachtube generated by PED for $\auto{Powertrn}$ is more precise, but for the rest, the reachtubes and the verification times using both GED and PED were comparable.  
%
%
%
%
%
%
%
In addition to the $\mathit{VerifySafety}$ algorithm, \toolname\ also
looks for counter-examples by quickly generating random executions of
the hybrid system.  If any of these executions is found to be unsafe,
\toolname\ will return ``Unsafe'' without starting the
$\mathit{VerifySafety}$ algorithm.



\section{Reasoning principles for trace containment}
\label{sec:trace-containment}

For a fixed unsafe set $\U$ and two hybrid systems $\H_1$ and $\H_2$,
proving $\reach{\H_1} \subseteq \reach{\H_2}$ and the safety of
$\H_2$, allows us to conclude the safety of
$\H_1$. Proposition~\ref{prop:hybrid-contain} establishes that proving
containment of traces, trajectories, and initial sets of two hybrid
systems, ensures the containment of their respective reach sets. These
two observations together give us a method of concluding the safety of
one system, from the safety of another, provided we can check trace
containment of two graphs, and trajectory containment of two
trajectory sets. In our examples, the set of modes $\L$ and the set of trajectories $\TL$ is often the same between the hybrid systems we care about. So in this section present different reasoning principles to check trace containment between two graphs.

Semantically, a transition graph $G$ can be viewed as one-clock timed automaton, i.e., one can constructed a timed automaton $T$ with
one-clock variable such that the timed traces of $T$ are exactly the
traces of $G$. This observation, coupled with the fact that checking
the timed language containment of one-clock timed
automata~\cite{ow-lics} is decidable, allows one to conclude that
checking if $G_1 \preceq_{\modemap} G_2$ is decidable. However the
algorithm in~\cite{ow-lics} has non-elementary complexity. Our next
observation establishes that forward simulation between graphs can be
checked in polynomial time. Combined with
Proposition~\ref{prop:graphsim}, this gives a simple sufficient condition for trace containment that can be efficiently checked.

\begin{proposition}
\proplabel{prop:check-sim}
Given graphs $G_1$ and $G_2$, and mode map $\modemap$, checking if
there is a forward simulation from $G_1$ to $G_2$ is in polynomial
time.
\end{proposition}

\begin{proof}
The result can be seen to follow from the algorithm for checking timed
simulations between timed automata~\cite{cerans} and the
correspondence between one-clock timed automata; the fact that the
automata have only one clock ensures that the region construction is
poly-sized as opposed to exponential-sized. However, in the special
case of transition graphs there is a more direct algorithm which does
not involve region construction that we describe here.

Observe that if $\{R_i\}_{i\in I}$ is a family of forward simulations
between $G_1$ and $G_2$ then $\cup_{i \in I} R_i$ is also a forward
simulation. Thus, like classical simulations, there is a unique
largest forward simulation between two graphs that is the greatest
fixpoint of a functional on relations over states of the transition
graph. Therefore, starting from the relation $\V_1 \times \V_2$, one
can progressively remove pairs $(v,u)$ such that $v$ is not simulated
by $u$, until a fixpoint is reached. Moreover, in this case, since
$G_1$ is a DAG, one can guarantee that the fixpoint will be reached in
$|\V_1|$ iterations.
\end{proof}

Executions of hybrid systems are for bounded time, and bounded number
of mode switches. This is because our transition graphs are acyclic
and the labels on edges are bounded intervals. Sequential composition
of graphs allows one to consider switching sequences that are longer
and of a longer duration. We now present observations that will allow
us to conclude the safety of a hybrid system with long switching
sequences based on the safety of the system under short switching
sequences. To do this we begin by observing simple properties about
sequential composition of graphs. In what follows, all hybrid systems
we consider will be over a fixed set of modes $\L$ and trajectory set
$\TL$. Also ${\sf id}$ will be identity function on $\L$.
Our first observation is that trace containment is consistent with
sequential composition.
\begin{proposition}
\label{prop:cong-seqcomp}
Let $G_i, G_i'$, $i \in \{1,2\}$, be four transition graphs over $\L$
such that $G_1\seqcomp G_2$ and $G_1'\seqcomp G_2'$ are defined, and
$G_i \preceq_{{\sf id}} G_i'$ for $i \in \{1,2\}$. Then $G_1\seqcomp
G_2 \preceq_{{\sf id}} G_1'\seqcomp G_2'$.
\end{proposition}

Next we observe that sequential composition of graphs satisfies the
``semi-group property''.
\begin{proposition}
\label{prop:semi-group}
Let $G_1,G_2$ be graphs over $\L$ for which $G_1\seqcomp G_2$ is
defined. Let $v_{1 {\sf term}}$ be the unique terminal vertex of
$G_1$. Consider the following hybrid systems: $\H = \langle \L,
\Theta, G_1\seqcomp G_2, \TL\rangle$, $\H_1 = \langle\L, \Theta, G_1,
\TL\rangle$, and $\H_2 = \langle\L, \reach{\H_1}^{v_{1 {\sf term}}}, $$
G_2,$ $ \TL\rangle$. Then
$\reach{\H} = \reach{\H_1} \cup \reach{\H_2}$.
\end{proposition}

Consider a graph $G$ such that $G\seqcomp G$ is defined. Let $\H$ be
the hybrid system with transition graph $G$, and $\H'$ be the hybrid
system with transition graph $G\seqcomp G$; the modes, trajectories,
and initial set for $\H$ and $\H'$ are the same. Now by
Proposition~\ref{prop:seqcomp-containment}
and~\ref{prop:hybrid-contain}, we can conclude that $\reach{\H}
\subseteq \reach{\H'}$. Our main result of this section is that under
some conditions, the converse also holds. This is useful because it
allows us to conclude the safety of $\H'$ from the safety of $\H$. In
other words, we can conclude the safety of a hybrid system for long, possibly unbounded, switching sequences (namely $\H'$) from the safety of the system under short switching sequences (namely $\H$).
\begin{theorem}
\label{thm:seqcomp-mainresult}
Suppose $G$ is such that $G\seqcomp G$ is defined. Let $v_{{\sf
    term}}$ be the unique terminal vertex of $G$. For natural number
$i \geq 1$, define $\H_i = \langle \L, \Theta, G^i, \TL\rangle$, where
$G^i$ is the $i$-fold sequential composition of $G$ with itself. In
particular, $\H_1 = \langle \L, \Theta, G, \TL\rangle$. If
$
\reach{\H_1}^{v_{{\sf term}}} \subseteq \Theta
$
then for all $i$, $\reach{\H_i} \subseteq \reach{\H_1}$.
\end{theorem}

\begin{proof}
Let $\Theta_1 = \reach{\H_1}^{v_{{\sf term}}}$. From the condition in
the theorem, we know that $\Theta_1 \subseteq \Theta$. Let us define
$\H_i' = \langle \L, \Theta_1, G^i, \TL\rangle$. Observe that from
Proposition~\ref{prop:hybrid-contain}, we have $\reach{\H_i'} \subseteq
\reach{\H_i}$.

The theorem is proved by induction on $i$. The base case (for $i = 1$)
trivially holds. For the induction step, assume that $\reach{\H_i}
\subseteq \reach{\H_1}$. Since $\seqcomp$ is associative, using
Proposition~\ref{prop:semi-group} and the induction hypothesis, we
have
$
\reach{\H_{i+1}} = \reach{\H_1} \cup \reach{\H_i'} \subseteq
    \reach{\H_1} \cup \reach{\H_i} = \reach{\H_1}.
$
\end{proof}

Theorem~\ref{thm:seqcomp-mainresult} allows one to determine the set
of reachable states of a set of modes $\L$ with respect to graph
$G^i$, provided $G$ satisfies the conditions in the statement. This
observation can be generalized. If a graph $G_2$ satisfies conditions
similar to those in Theorem~\ref{thm:seqcomp-mainresult}, then using
Proposition~\ref{prop:semi-group}, we can conclude that the reachable
set with respect to graph $G_1\seqcomp G_2^i\seqcomp G_3$ is contained
in the reachable set with respect to graph $G_1\seqcomp G_2\seqcomp
G_3$. The formal statement of this observation and its proof is
skipped in the interest of space, but we will use it in our
experiments.

\subsection{Experiments on trace containment reasoning}
\label{sec: sub_exp}

\paragraph{Graph simulation}
Consider the $\auto{AEB}$ system of Section~\ref{ssec:adas} with the
scenario where Vehicle B is stopped ahead of vehicle A, and A transits
from \Lmode{cruise} to \Lmode{em\_brake} to avoid colliding with B.
In the actual system ($G_2$ of Figure~\ref{fig: em_brake_graphs}), two different sensor systems
trigger the obstacle detection and emergency braking at time intervals
$[1,2]$ and $[2.5,3.5]$ and take the system from vertex $0$
(\Lmode{cruise}) to two different vertices labeled with
\Lmode{em\_brake}.

\begin{figure}[ht]
    \begin{subfigure}[b]{0.3\textwidth}
        \includegraphics[width=\textwidth]{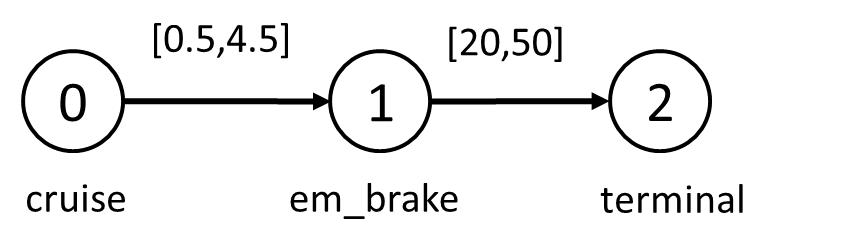}
        \caption{\small Transition graph $G_1$.}
        \label{fig: em_brake_graphsA}
    \end{subfigure}
    \begin{subfigure}[b]{0.3\textwidth}
        \includegraphics[width=\textwidth]{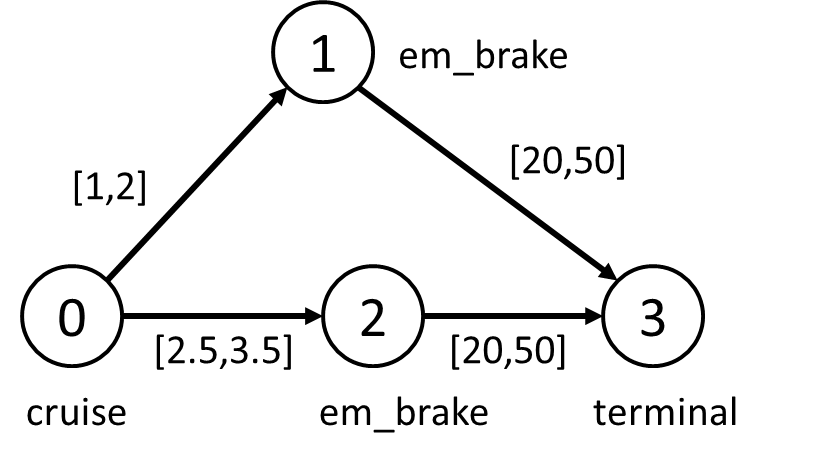}
        \caption{Transition graph $G_2$.}
        \label{fig: em_brake_graphsB}
    \end{subfigure}
    \begin{subfigure}[b]{0.35\textwidth}
	\includegraphics[width=\textwidth]{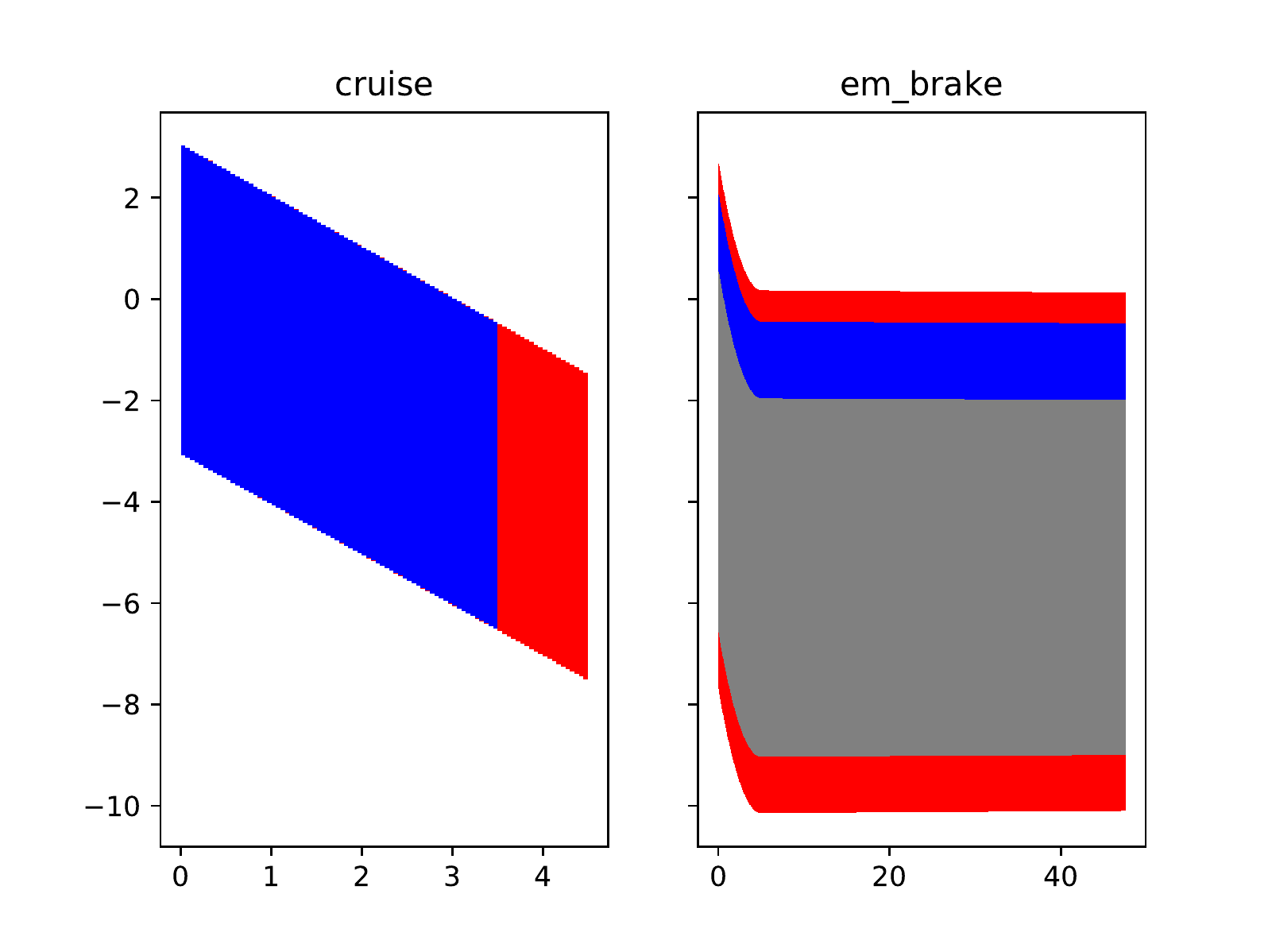}
	\caption{$\auto{AEB}$ Reachtubes.}
	\label{fig:abstraction}
    \end{subfigure}
\caption{\small Graphs and reachtubes for the Automatic Emergency Braking $\auto{AEB}$ system.}
\label{fig: em_brake_graphs}
\end{figure}

To illustrate trace containment reasoning, consider a simpler
graph $G_1$ that allows a single transition of A from \Lmode{cruise}
to \Lmode{em\_brake} over the interval bigger $[0.5, 4.5]$. Using
Proposition~\ref{prop:hybrid-contain} and checking that graph $G_2
\preceq_{\sf id} G_1$, it follows that verifying the safety of
$\auto{AEB}$ with $G_1$ is adequate to infer the safety with $G_2$.
Figure~\ref{fig:abstraction} shows that the safe reachtubes returned
by the algorithm for $G_1$ in red, indeed contain the reachtubes for
$G_2$ (in blue and gray).




\vspace{-10pt}
\paragraph{Sequential composition}
We revisit the $\auto{Powertrn}$ example of
Section~\ref{ex:powertrain}. The initial set $\Theta$ and unsafe set
are the same as in Table \ref{table:results}.  Let $G_A$ be the graph
($v_0$,\Lmode{startup}) $\xrightarrow{[5,10]}$ ($v_1$,\Lmode{normal})
$\xrightarrow{[10,15]}$ ($v_2$,\Lmode{powerup}), and $G_B$ be the
graph ($v_0$,\Lmode{powerup}) $\xrightarrow{[5,10]}$
($v_1$,\Lmode{normal}) $\xrightarrow{[10,15]}$
($v_2$,\Lmode{powerup}).  The graph $G_1 =$ ($v_0$,\Lmode{startup})
$\xrightarrow{[5,10]}$ ($v_1$,\Lmode{normal}) $\xrightarrow{[10,15]}$
($v_2$,\Lmode{powerup}) $\xrightarrow{[5,10]}$ ($v_3$,\Lmode{normal})
$\xrightarrow{[10,15]}$ ($v_4, $\Lmode{powerup}), can be expressed as
the composition $G_1 = G_A \seqcomp G_B$. Consider the two hybrid
systems $\H_i = \langle \L, \Theta_i, G_i, \TL \rangle$, $i \in
\{A,B\}$ with $\Theta_A = \Theta$ and $\Theta_B =
\reach{\H_A}^{v_2}$. {\toolname}'s estimate of $\Theta_{B}$ had
$\lambda$ in the range from $14.68$ to $14.71$. The reachset
$\reach{\H_B}^{v_2}$ computed by {\toolname} had $\lambda$ from
$14.69$ to $14.70$. The remaining variables also were observed to
satisfy the containment condition. Therefore, $\reach{\H_B}^{v_2}
\subseteq \Theta_{B}$.  Consider the two hybrid systems $\H_i =
\langle \L, \Theta, G_i, \TL \rangle$, $i \in \{1,2\}$, where $G_1$ is
(defined above) $ G_A \seqcomp G_B$, and $G_2 = G_A \seqcomp G_B
\seqcomp G_B \seqcomp G_B$.  Using Theorem
\ref{thm:seqcomp-mainresult} it suffices to analyze $\H_1$ to verify
$\H_2$.  $\H_1$ was been proved to be safe by \toolname\ without any
refinement.  As a sanity check, we also verified the safety of $\H_2$.
{\toolname} proved $\H_2$ safe without any refinement as well.

\section{Conclusions}
\label{sec:cs}

The work presented in this paper takes an alternative view that complete mathematical models of hybrid systems are unavailable. Instead, the available system description combines a black-box simulator and a white-box transition graph.
Starting from this point of view, we have developed the semantic framework, a probabilistic verification algorithm, and results on
simulation relations and sequential composition for reasoning about
complex hybrid systems over long switching sequences.
Through modeling and analysis of a number of automotive control systems using implementations of the proposed approach, we hope to have demonstrated their promise.
One direction for further exploration in this vein,
is to consider more general timed and hybrid automata models of the white-box, and develop the necessary algorithms and the reasoning techniques.  

\newpage

\bibliographystyle{plain}
\bibliography{sayan,mahesh}

\newpage

\appendix
\section{Appendix}
\label{app:A}

\subsection{ADAS and autonomous vehicle venchmarks}
\label{app:adas}
We provide more details for the different scenarios used for testing ADAS and Autonomous driving control systems.

Recall, that each vehicle model in \Simulink\ has several continuous variables including the  $x, y$-coordinates of the vehicle on the road, its velocity, heading, steering angle, etc. 
The vehicle can be controlled by two input signals, namely the throttle (acceleration or brake) and the steering speed.
By choosing appropriate values of these input signals, we have defined the following  modes for each vehicle
\begin{inparaenum}[(a)]
	\item \Lmode{cruise}: move forward at constant speed, 
	\Lmode{speedup}: constant acceleration,
	\Lmode{brake}: constant (slow) deceleration,
	\Lmode{em\_brake}: constant (hard).
\end{inparaenum}
We have designed lane switching modes \Lmode{ch\_left} and \Lmode{ch\_right} in which the acceleration and steering are controlled in such a manner that the vehicle switches to its left (resp. right) lane in a certain amount of time. 

For each vehicle, we mainly analyze four variables: absolute position
($sx$) and velocity $vx$ orthogonal to the road direction ($x$-axis),
and absolute position ($sy$) and velocity $vy$ along the road
direction ($x$-axis). The throttle and steering information can be
expressed using the four variables. We will use subscripts to
distinguish between different vehicles.  The following scenarios are
constructed by defining appropriate sets of initial states and
transitions graphs labeled by the modes of two or more vehicles.
\begin{description} 
	\item[$\auto{MergeBehind}$:] 
	Initial condition:  Vehicle A is in left and  vehicle B is in the right lane; initial positions and speeds are in some range;  A is in \Lmode{cruise} mode, and B is in \Lmode{cruise} or \Lmode{speedup}.
	Transition graph:  Vehicle A  goes through the mode sequence \Lmode{speedup}, \Lmode{ch\_right}, \Lmode{cruise} with specified intervals of time to transit from mode to another mode. 
	Requirement: A merges behind B within a time bound and maintains at least a given safe separation.
	\item[$\auto{MergeAhead}$:] 
	Initial condition: Same as  $\auto{MergeBehind}$ with 
	except that B is in \Lmode{cruise} or \Lmode{brake} mode.
	Transition graph: Same structure as  $\auto{MergeBehind}$ with different  timing parameters.
	Requirement: A merges ahead of B and maintains at least a given safe separation. 
	\item[$\auto{AutoPassing}$:]
	Initial condition: Vehicle A behind  B in the same lane, with A in \Lmode{speedup} and B in \Lmode{cruise}; initial positions and speeds are in some range.
	Transition graph:  A goes through the mode sequence \Lmode{ch\_left}, \Lmode{speedup}, \Lmode{brake}, and  \Lmode{ch\_right}, \Lmode{cruise} with specified time intervals in each mode to complete the overtake maneuver. If B switches to
	\Lmode{speedup} before A enters \Lmode{speedup} then
	A aborts and changes back to right lane. If B switches to \Lmode{brake} before A enters \Lmode{ch\_left}, then A should adjust the time to switch to \Lmode{ch\_left} to avoid collision.
	%
	Requirement: Vehicle A overtakes B while maintaining minimal safe separation.
	
	\item[$\auto{AEB}$:]
	(Emergency brakes) 
	Initial condition: Vehicle A behind  B in the same lane with A in \Lmode{cruise}, B is stopped (in \Lmode{cruise} mode with velocity $0$). Initial positions and speeds are in some range;
	Transition graph: A transits from \Lmode{cruise} to \Lmode{em\_brake}  over a given interval of time or several disjoint intervals of time.
	Requirement: Vehicle A stops behind B and maintains at least a given safe separation.
	
	\item[$\auto{MergeBetween}$:]
	Initial condition: Vehicle A, B, C are all in the same lane, with A behind B, B behind C, and in the \Lmode{cruise} mode, initial positions and speeds are in some range.
	Transition graph: A goes through the mode sequence \Lmode{ch\_left}, \Lmode{speedup}, \Lmode{brake}, and  \Lmode{ch\_right}, \Lmode{cruise} with specified time intervals in each mode to overtake B. C transits from \Lmode{cruise} to \Lmode{speedup} then transits back to \Lmode{cruise}, so C is always ahead of A.
	Requirement: Vehicle A merges between B and C and any two vehicles maintain at least a given safe separation.
\end{description}

\subsection{Automatic transmission control}
\label{ssec:gear}
We provide some details about the Automatic transmission control benchmark that we have modeled as a hybrid system that combine white-box and black-box components and we have verified using \toolname's safety verification algorithm.

This is a slightly modified version of the Automatic Transmission model provided by \Mathworks\ as a \Simulink\ demo \cite{Matlab_trans}. 
It is a model of an automatic transmission controller that exhibits both continuous and discrete behavior.
The model has  been previously  used by  S-taliro~\cite{S-Taliro} for falsifying certain requirements. We are not aware of any verification results for this system.

For our experiments, we made some minor modifications to  the \Simulink\ model to create the hybrid system  $\auto{ATS}$. 
This allows us to simulate the vehicle from any one of the four modes, namely,  \Lmode{gear1}, \Lmode{gear2}, \Lmode{gear3} and \Lmode{gear4}.
Although the system has many variables, we are primarily  interested in the car Speed ($v$), engine RPM (Erpm), impeller torque ($T_i$), output torque ($T_o$), and transmission RPM (Trpm), and therefore, use simulations that record these. 
Transition graph of  $\auto{ATS}$ encodes 
transition sequences and intervals for
shifting from \Lmode{gear1} through to \Lmode{gear4}.
Requirement of interest is that the engine RPM is less than a specified maximum value, which in turn is important for limiting the thermal and mechanical stresses on the cylinders and camshafts. Typical unsafe set $\U_t$ could be  Erpm $>4000$.
%

%
%

\subsection{Safety verification algorithm}
\label{appendix:safetyveri}

The safety verification algorithm is shown in~\ref{alg:safetyveri}. It proceeds along the line of the simulation-based verification algorithms presented in~\cite{DMV:EMSOFT2013, FanMitra:2015df, DuggiralaMV:2015c2e2}.

\begin{algorithm}
	\caption{$\mathit{VerifySafety}(\H,\U)$ verifies safety of hybrid system $\H$ with respect to unsafe set $\U$.}
	\label{alg:safetyveri}
	\SetKwInOut{Input}{input}
	\SetKwInOut{Initially}{initially}
	\Initially{$\I.push(Partition(\Theta))$}
	\While {$\I \neq \emptyset$}
	{	
		$S \gets \I.pop()$\;
		$RS \gets \GraphReach(\H)$ \;
		\uIf {$RS \cap \U = \emptyset$}
		{ continue\;}
		\uElseIf {$\exists (x,l,t) \in RT$ s.t. $\langle RT,v \rangle \in RS$ and $(x,l,t) \subseteq \U $}
		{\Return UNSAFE, $\langle RT,v \rangle$}
		\Else
		{
			{$I.push (Partition(S))$ \;}	
			{Or, $G \gets RefineGraph(G)$ \;}
		}
	}
	\Return SAFE
\end{algorithm}

\end{document}